\documentclass[11pt,american,twoside,a4paper]{article}

\usepackage[T1]{fontenc}
\usepackage[utf8]{inputenc}
\usepackage[american,provide=*]{babel}
\usepackage{a4wide}
\usepackage{fourier,bbm,bm}
\usepackage{BOONDOX-calo}
\usepackage{amsmath,amsfonts,amsthm,amssymb}
\usepackage{color,xcolor,graphicx}
\usepackage{latexsym,mathtools,leftindex}
\usepackage[shortlabels]{enumitem}
\usepackage{xargs,xspace,needspace,xparse,easyReview}
\usepackage[active]{srcltx}
\usepackage[colorlinks=true]{hyperref}
\usepackage{fancyhdr}
\usepackage[toc]{appendix}
\usepackage[normalem]{ulem}
\usepackage[style=cap-alpha]{biblatex}

\colorlet{darkblue}{blue!50!black}
\hypersetup{
	colorlinks,%
	citecolor=darkblue,%
	filecolor=red,%
	linkcolor=darkblue,%
	urlcolor=darkblue,%
	pdfnewwindow=true,%
	pdfstartview={FitH}
}

\numberwithin{equation}{section}
\theoremstyle{plain}
\newtheorem{theorem}{Theorem}[section]
\AtBeginEnvironment{theorem}{\Needspace{5\baselineskip}}
\newcommand\bet{\begin{theorem}} \newcommand\eet{\end{theorem}}
\newtheorem{proposition}[theorem]{Proposition}
\AtBeginEnvironment{proposition}{\Needspace{5\baselineskip}}
\newcommand\bep{\begin{proposition}} \newcommand\eep{\end{proposition}}
\newtheorem{lemma}[theorem]{Lemma}
\AtBeginEnvironment{lemma}{\Needspace{5\baselineskip}}
\newcommand\bel{\begin{lemma}} \newcommand\eel{\end{lemma}}
\newtheorem{corollary}[theorem]{Corollary}
\AtBeginEnvironment{corollary}{\Needspace{5\baselineskip}}
\newcommand\bec{\begin{corollary}} \newcommand\eec{\end{corollary}}
\theoremstyle{definition}
\newtheorem{definition}[theorem]{Definition}
\AtBeginEnvironment{definition}{\Needspace{5\baselineskip}}
\newcommand\bed{\begin{definition}} \newcommand\eed{\end{definition}}
\newtheorem{remark}{Remark}[section]
\AtBeginEnvironment{remark}{\Needspace{5\baselineskip}}
\newcommand\ber{\begin{remark}} \newcommand\eer{\end{remark}}
\newtheorem{remarks}[remark]{Remarks}
\AtBeginEnvironment{remarks}{\Needspace{5\baselineskip}}
\newcommand\bers{\begin{remarks}} \newcommand\eers{\end{remarks}}
\AtBeginEnvironment{quote}{\Needspace{5\baselineskip}}

\newcommand\beq{\begin{equation}} \newcommand\eeq{\end{equation}}

\newcounter{smallarabics}
\newenvironment{arabicenumerate}
{\begin{list}{{\normalfont\textrm{(\arabic{smallarabics})}}}
		{\usecounter{smallarabics}\setlength{\itemindent}{0cm}
			\setlength{\leftmargin}{5ex}\setlength{\labelwidth}{4ex}
			\setlength{\topsep}{0.75\parsep}\setlength{\partopsep}{0ex}
			\setlength{\itemsep}{0ex}}}
	{\end{list}}
\newcommand{\ben}{\begin{arabicenumerate}} \newcommand{\een}{\end{arabicenumerate}}

\newcounter{smallroman}
\newenvironment{romanenumerate}
{\begin{list}{{\normalfont\textbf{(\roman{smallroman})}}}
		{\usecounter{smallroman}\setlength{\itemindent}{0cm}
			\setlength{\leftmargin}{5ex}\setlength{\labelwidth}{4ex}
			\setlength{\topsep}{0.75\parsep}\setlength{\partopsep}{0ex}
			\setlength{\itemsep}{0ex}}}
	{\end{list}}
\newcommand{\benr}{\begin{romanenumerate}} \newcommand{\eenr}{\end{romanenumerate}}

\AtBeginEnvironment{theorem}{\Needspace{5\baselineskip}}
\AtBeginEnvironment{proposition}{\Needspace{3\baselineskip}}
\AtBeginEnvironment{definition}{\Needspace{5\baselineskip}}
\AtBeginEnvironment{corollary}{\Needspace{5\baselineskip}}
\AtBeginEnvironment{lemma}{\Needspace{5\baselineskip}}
\AtBeginEnvironment{quote}{\Needspace{5\baselineskip}}

\NewDocumentCommand{\ASS}{mm}{\expandafter\newcommand\csname #1\endcsname{{\hyperref[#1]{\bf (#2)}}}}
\NewDocumentCommand{\preASS}{mm}{\expandafter\newcommand\csname pre#1\endcsname{{\hyperref[#1]{\bf (#2)}}}}

 \newcommand\cB{{\mathcal B}} 
  \newcommand\cF{{\mathcal F}}

\newcommand\cS{{\mathcal S}}  
 \newcommand\cW{{\mathcal W}}

 \newcommand\ce{{\mathcal e}}

\newcommand\ZZ{{\mathbbm Z}}

\newcommand\RR{{\mathbbm R}}
\newcommand\CC{{\mathbbm C}}

\newcommand\TT{{\mathbbm T}}

\newcommand{\one}{\mathbbm{1}}

\newcommand\fA{{\mathfrak A}}

 \newcommand\fh{{\mathfrak h}}

\newcommand\e{\mathrm{e}}

\renewcommand{\i}{\mathrm{i}}
\renewcommand\d{\mathrm{d}}

\newcommand{\Id}{\mathrm{Id}}

\newcommand\myoperator[1]{\mathop{\mathrm{#1}}}

\newcommand\Dom{\myoperator{Dom}\nolimits}

\renewcommand\Im{\myoperator{Im}\nolimits}
\newcommand\spec{\myoperator{sp}\nolimits}
\newcommand\sgn{\myoperator{sgn}\nolimits}

\renewcommand\det{\myoperator{det}\nolimits}
\newcommand\tr{\myoperator{tr}\nolimits}

\newcommand\supp{\myoperator{supp}\nolimits}

\newcommand{\CAR}{\myoperator{CAR}\nolimits}
\newcommand\Ent{\myoperator{Ent}\nolimits}

\renewcommand\proof{\noindent{\bf Proof.\ }}
\renewcommand\qed{\hfill$\Box$}
\newcommand\ie{\textsl{i.e.,\ }}
\renewcommand\bar{\overline}

\newcommand{\twol}[2]{\genfrac{}{}{0pt}{1}{#1}{#2}}

\newcommand\eq{\mathrm{eq}}

\newcommand\loc{\mathrm{loc}}

\newcommand\Lim{\lim_\Lambda}
\newcommand\fAloc{\fA_\loc}

\newcommand\cSI{\cS_\mathrm{I}}
\newcommand\cSeq{\cS_\mathrm{eq}}

\newcommand\cSKMS{\cS_\mathrm{KMS}}

\newcommand{\un}{\mathbf{1}}
\newcommand{\hdec}{{h_\Lambda^{\mathrm{dec}}}}
\newcommand{\Phidec}{{\Phi_\Lambda^{\mathrm{dec}}}}

\begin{document}
\def\today{}

\title{Entropy Maximization and Weak Gibbsianity of Quasi-Free Fermionic States}

\bigskip
\author{Vojkan Jak\v{s}i\'c$^{1}$\quad Claude-Alain Pillet$^2$\quad Anna Szczepanek$^3$\\ 
\\ \\
$^{1}$Dipartimento di Matematica, Politecnico di Milano\\
piazza Leonardo da Vinci, 32 \\
20133 Milano,  Italy
\\ \\ 
$^2$Universit\'e de Toulon, CNRS, CPT, UMR 7332, 83957 La Garde, France\\
Aix-Marseille Univ, CNRS, CPT, UMR 7332, Case 907, 13288 Marseille, France 
\\ \\
$^3$Institute of Mathematics, Jagiellonian University,\\ 
{\L}ojasiewicza 6, 30-348 Krak\'ow, Poland
}
\maketitle
\thispagestyle{empty}
\bigskip
\bigskip
\noindent{\small{\bf Abstract.} 
In their 1972 study of {\sl approach to equilibrium,} Lanford and Robinson
showed that gauge-invariant quasi-free states of lattice fermions maximize
entropy among all translation-invariant states with a fixed two-point function 
and suggested that the maximizer is unique. In subsequent work on this topic,
the uniqueness question re-emerged, together with the problem of whether such
quasi-free states are weak Gibbs states. We provide a positive answer to both
questions within a class of states whose momentum-space two-point function 
$\widehat C$ satisfies $0<\widehat C(k)<1$ and belongs to the Wiener algebra
of the Brillouin zone. The proof reveals that both the entropy maximization
principle and weak Gibbsianity follow directly from the thermodynamic formalism
for lattice fermions.}

\tableofcontents

\section{Introduction}

In their 1972 study of approach to equilibrium of free fermionic
systems~\cite{Lanford1972}, Lanford and Robinson proved, under natural
thermodynamic assumptions, that among all translation-invariant states with
fixed two-point function, the specific entropy is maximized by gauge-invariant
quasi-free states. They further suggested that this maximizer should be unique;
see Proposition~1a and the end of Section~A in~\cite{Lanford1972}.

We establish this conjecture for systems on the lattice $\ZZ^d$, under a
natural regularity assumption on their two-point function. We also prove the
related conjecture that these maximizers are weak Gibbs states. Both
results follow directly from the thermodynamic formalism for lattice fermions
developed by Araki--Moriya~\cite{Araki2003} and from the perspective on weak
Gibbsianity introduced in~\cite{Jaksic2023a,Jaksic2026b}. These
observations fit naturally into the research program initiated in the latter
works. Since the arguments are conceptually simple and may be of independent
interest, we present them here in a self-contained form.

To state the main result, we first recall the standard CAR framework and the
Lanford--Robinson result~\cite[Proposition~1a]{Lanford1972}; further details are
provided in Section~\ref{sec-prelim}.

The one-particle space for fermions is $\fh=\ell^2(\ZZ^d)$, and we denote by
$\fh_X=\ell^2(X)$ the closed subspaces of $\fh$ associated to subsets $X\subset\ZZ^d$. 
We let $\left(\delta_x\right)_{x\in\ZZ^d}$ be the canonical orthonormal basis of $\fh$. 
We further denote by $\hat f\in L^2(\TT^d,\tfrac{\d	k}{(2\pi)^d})$ the Fourier transform 
of $f\in\fh$,\footnote{$\TT^d$ denotes the Brillouin zone $[-\pi,\pi]^d$, identified 
with the $d$-dimensional torus.} and by $\cB(\fh)$ the algebra of bounded operators on
$\fh$.

$\fA=\CAR(\fh)$ denotes the CAR algebra over $\fh$, and $a(f)$/$a^\ast(f)$ the 
annihilation/creation operator associated with $f\in\fh$. For $x\in\ZZ^d$, 
we abbreviate $a_x\coloneq a(\delta_x)$ and $a_x^\ast\coloneq a^\ast(\delta_x)$.
$\cSI$ denotes the set of all translation-invariant states on $\fA$.
To $\omega\in\cSI$ we associate its specific entropy $s(\omega)$
and its covariance, the operator $C\in\cB(\fh)$ defined by
$$
\langle f,Cg\rangle=\omega(a^\ast(g)a(f)).
$$
$C$ commutes with the unitary action of $\ZZ^d$ and satisfies $0\le C\le\one$.
We shall say that such a $C\in\cB(\fh)$ is a \textit{covariance operator.}

To a covariance operator $C\in\cB(\fh)$ we associate its \textit{two-point function}
$C(x,y)=\langle\delta_x,C\delta_y\rangle=(C\delta_y)(x)$. In the Fourier representation of $\fh$,
$C$ acts as a multiplication operator
$$
(Cf)\hat{\ }(k)=\widehat{C}(k)\hat f(k),
$$
where $\TT^d\ni k\mapsto\widehat{C}(k)$ is the Fourier transform of $C\delta_0\in\fh$, 
and hence a measurable function  satisfying $\widehat{C}(k)\in[0,1]$ 
for a.e.\;$k\in\TT^d$. We denote by $\omega_C$ the unique translation-invariant 
and \textit{gauge-invariant quasi-free state} on $\fA$ with covariance $C$. Its specific 
entropy is given by 
\[
s(\omega_C)=-\int_{\TT^d}\left(\widehat{C}(k)\log\widehat{C}(k)
+(1-\widehat{C}(k))\log(1-\widehat{C}(k))\right)\frac{\d k}{(2\pi)^d}.
\]
Finally, $\cSI(C)$ denotes the set of all $\omega\in\cSI$ with covariance $C$. 
Clearly $\omega_C\in\cSI(C)$. Lanford and Robinson prove
the following result~\cite[Proposition 1a]{Lanford1972}.
\begin{theorem} \label{thm:lr}
\beq
\sup _{\omega \in \cSI(C)}s(\omega)= s(\omega_C)
\label{eq:lr}
\eeq
\end{theorem}
In the same work they have also suggested that $\omega_C$ is the unique maximizer 
in~\eqref{eq:lr}. In other words, one expects that if the two-point correlations
are fixed and maximal disorder is imposed, then no additional structure survives: 
all higher-order correlations are completely determined by the two-point function 
through the fermionic Wick formula.

The space $\ell^1(\ZZ^d)$ equipped with the convolution operation 
$$
(f\star g)(x)=\sum_{y\in\ZZ^d}f(x-y)g(y),
$$
is a
Banach $\ast$-algebra. The Fourier transform maps it onto a subalgebra of $C(\TT^d)$, 
the Wiener algebra $\cW$. We will need the following notion of regularity for
covariance operators.
\begin{definition}\label{def:tr-C}
A covariance operator $C\in\cB(\fh)$ is called regular if $\widehat{C}\in\cW$ and
$\widehat{C}(\TT^d)\!\subset\,]0,1[$.
\end{definition}

Note that $\widehat{C}\in\cW$ iff the function $x\mapsto C(x,0)$ belongs to $\ell^1(\ZZ^d)$.
Let us explore the consequences of this definition from the perspective of the 
thermodynamic formalism. The function $m(z)=\log((1-z)/z)$ is analytic on the
cut plane $\CC\setminus\big(\,]-\infty,0]\cup[1,\infty[\,\big)$. Since $\widehat{C}$
maps $\TT^d$ to a compact subset of $]0,1[$, Wiener--Levy's theorem implies that 
$\widehat{h}\coloneq m\circ\widehat{C}\in\cW$. We denote by $h$ the
associated operator on $\fh$, and note that since $m\colon ]0,1[\to\RR$,
$h$ is a self-adjoint element of $\cB(\fh)$ such that
$$
C=\left(\one+\e^h\right)^{-1}.
$$
We will refer to $h$ as the one-particle Hamiltonian associated to $C$
and denote by $\tau_h=\left(\tau_h^t\right)_{t\in\RR}$ the group of 
Bogoliubov automorphisms of $\fA$ generated by $h$. The unique $\tau_h$-KMS
state at $\beta=1$ is $\omega_C$. It also
follows from the Araki--Moriya thermodynamic formalism of fermionic 
systems that $\omega_C$ is the unique equilibrium state for the 
fermionic interaction 
$\Phi$ defined by\footnote{Here, one allows for $x=y$.} 
\beq
\Phi(X)\coloneq\begin{cases*}
\tfrac12\left(h(x,y)a_x^\ast a_y+h(y,x)a_y^\ast a_x\right)
& if $X = \{x,y\}$;\\[4pt]
0&otherwise,
\end{cases*}
\label{eq:Phidef}
\eeq
where $h(x,y)=\langle\delta_x,h\delta_y\rangle$. 
We will review the thermodynamic formalism of fermionic systems  
and discuss the naturalness of definition~\ref{def:tr-C} in this context
in Section~\ref{sec-prelim}.

\newpage

Our first result is 
\begin{theorem}\label{thm:main}
Suppose that $C$ is a regular covariance operator. Then $\omega_C$ is the 
unique maximizer in~\eqref{eq:lr}.
\end{theorem}
 
We now turn to our second result. In the following, $\Lambda$ always denotes
a cubic box in $\ZZ^d$ centered at the origin, and we write $\Lim$ for a 
limit over an increasing sequence of such cubes. The local Hamiltonian 
induced in the box $\Lambda$ by the interaction~\eqref{eq:Phidef} is
\[
H_\Lambda(\Phi)\coloneq\sum_{X\subset\Lambda}\Phi(X)
=\sum_{x,y\in\Lambda}h(x,y)a_x^\ast a_y.
\]
It coincides with the gauge-invariant quadratic Hamiltonian associated 
to the compression of $h$ to the finite-dimensional subspace $\fh_\Lambda$,
$h_\Lambda\coloneq\un_\Lambda h\un_\Lambda$, where 
$\un_\Lambda$ denotes the orthogonal projection on $\fh_\Lambda$.

We denote by $\omega_\Lambda$ 
the restriction of a state $\omega \in \cSI$ to the finite-dimensional subalgebra
$\CAR(\fh_\Lambda)$. We identify $\omega_\Lambda$ with its density matrix in 
the canonical trace representation. The pressure of the interaction~\eqref{eq:Phidef}
in the box $\Lambda$ is
$$
P_\Lambda(\Phi)=\log\tr(\e^{-H_\Lambda(\Phi)}).
$$

\begin{theorem}\label{thm:weak-gibbs} With the same hypothesis as in Theorem \ref{thm:main}, 
the state $\omega_C$ is weak Gibbs in the sense that there are positive constants $c_\Lambda$ 
satisfying 
\[
\Lim\tfrac{\log c_\Lambda}{|\Lambda|}=0,
\]
and 
$$
c_\Lambda^{-1}\e^{- H_\Lambda(\Phi)-P_\Lambda(\Phi)}
\leq \omega_{C,\Lambda}
\leq c_\Lambda \e^{-H_\Lambda(\Phi)-P_\Lambda(\Phi)}.
$$
\end{theorem}

Weak Gibbsianity plays a  central role in the research program introduced in~\cite{Jaksic2023a,Jaksic2026b}. 
 
To our knowledge, it has not been explicitly observed in the literature that
Theorem~\ref{thm:main} follows directly from the non-negativity of relative
entropy density\footnote{This fact also plays a central role in Lanford--Robinson's
proof of Theorem~\ref{thm:lr}.} and the Araki--Moriya thermodynamic formalism. 
Likewise, Theorem~\ref{thm:weak-gibbs} is a consequence of~\cite[Theorem~3.1]{Jaksic2023a}. 
Gauge-invariant quasi-free states play an important role across physics and mathematical 
physics, and we hope that these observations may find applications beyond the study of
approach to equilibrium.

Finally, we remark that Lanford and Robinson in~\cite{Lanford1972} worked in
the continuous Fermi systems setting in which the one-particle Hilbert space is
$L^2(\RR^d)$. The thermodynamic formalism of such systems has not been
developed, and our method of proof does not apply.

\paragraph*{Acknowledgments.} This work was partly funded by the CY
Initiative grant Investisse\-ments d’Avenir, grant number ANR-16-IDEX-0008. VJ
acknowledges the support of the MUR grant ``Dipartimento di Eccellenza 2023--2027''
of Dipartimento di Matematica, Politecnico di Milano. We also acknowledge the
support of the ANR project DYNACQUS, grant number ANR-24-CE40-5714. 

The results of this note were developed during the Winter 2026 Mathematical
Physics Working Seminar at the Politecnico di Milano. We thank T.~Benoist,
A.~Lucia, M.~Moscolari, C.~Vogel, and M.~Wrochna for valuable discussions.

 \section{Lattice Fermions}
 \label{sec-prelim}

 \subsection{Quasi-Free States}

We describe gauge-invariant quasi-free states of the CAR algebra over 
an arbitrary Hilbert space $\fh$.

$\CAR(\fh)$ is the $C^\ast$-algebra generated by the creation/annihilation 
operators $a^\ast(f)/a(f)$ associated to elements $f\in\fh$, satisfying
the anti-commutation relations
$$
\{a(f),a^\ast(g)\}=\langle f,g\rangle,\qquad\{a(f),a(g)\}=0,
$$
for all $f,g\in\fh$. Let $h$ be a (possibly unbounded) self-adjoint operator on $\fh$. 
The  maps
\[
\tau_h^t(a(f))= a(\e^{\i th}f)
\]
uniquely extend to a strongly continuous group $\left(\tau_h^t\right)_{t\in\RR}$
of $\ast$-automorphisms of $\CAR(\fh)$, the group of \textit{Bogoliubov automorphisms} 
generated by $h$. It describes the Heisenberg dynamics of a system of non-interacting 
fermions with one-particle Hamiltonian $h$.

More generally, a $C^\ast$\textit{-dynamics} on $\CAR(\fh)$ is a strongly continuous group 
$\alpha=\left(\alpha^t\right)_{t\in\RR}$ of $\ast$-auto\-mor\-phisms. A state $\omega$ on 
$\CAR(\fh)$ is said to be \textit{KMS for} $\alpha$ \textit{at inverse temperature}
$\beta\neq0$ whenever, for any $A,B\in\CAR(\fh)$, there exists a function $F_{A,B}$
which is analytic in the strip $\big\{z\in\CC\mid \sgn(\beta)\operatorname{Im}(z)\! \in \,]0,|\beta|[\big\}$ and continuous 
on its closure such that, for t$\in\RR$,
$$
F_{A,B}(t)=\omega(A\alpha^t(B)),\qquad 
F_{A,B}(t+\i\beta)=\omega(\alpha^t(B)A).
$$
The \textit{gauge group} of $\CAR(\fh)$ is the group of Bogoliubov automorphisms
generated by the unit $\one$, \ie $\tau_\one^\theta(a(f))=a(\e^{\i\theta}f)$. A state
$\omega$ on $\CAR(\fh)$ is \textit{gauge-invariant} whenever
$\omega\circ\tau_\one^\theta=\omega$ for all $\theta\in\RR$.

A gauge-invariant state $\omega$ with covariance operator $C$ is \textit{quasi-free} iff,
for all integers $n,m$ and any $f_1,\ldots,g_m\in\fh$,
\[
\omega(a^\ast(f_1)\cdots a^\ast(f_n)a(g_m)\cdots a(g_1))
=\delta_{nm}\det\left(
\left[\langle g_i,Cf_j\rangle\right]_{i,j\in\llbracket1,n\rrbracket}
\right).
\]
Any covariance operator $C$ determines a unique gauge-invariant quasi-free 
state $\omega_C$ by the same formula. This state  is faithful iff 
$\spec(C)\cap\{0,1\}=\emptyset$.\footnote{$\spec(C)$ denotes the spectrum of $C$.}

A faithful gauge-invariant quasi-free state with the covariance $C$ is a KMS state 
at $\beta=1$ for the group of Bogoliubov automorphisms generated by $h=\log((\one-C)C^{-1})$.
Conversely, 
\begin{theorem}\label{thm:kms-qf}
Let $h$ be a self-adjoint operator on $\fh$. The  gauge-invariant quasi-free state with covariance  
$$
C=\left(\one+\e^h\right)^{-1}
$$
is the unique $\beta=1$ KMS state  for $\tau_h$ on $\CAR(\fh)$. 
\end{theorem}
For the proof of this result, see \cite[Example 5.3.2]{Bratteli1981}.

\subsection{Thermodynamic Formalism}

The thermodynamic formalism of lattice quantum spin systems, developed in the
1960s and 1970s, is discussed in detail in the classical 
monographs~\cite{Israel1979,Bratteli1981,Simon1993}. Its extension to lattice fermions, 
placing them on the same footing as spin systems, was developed in the major 
work~\cite{Araki2003}. In this section we collect a number of results 
from~\cite{Araki2003} that play a role in the proof of Theorem~\ref{thm:main}.

As in the introduction, we set $\fh=\ell^2(\ZZ^d)$ and $\fA=\CAR(\fh)$. 
$\cF$ denotes the set of all finite subsets of $\ZZ^d$, and $|X|$ the cardinality 
of $X\in\cF$. To each $X\in\cF$ we associate local algebra $\fA_X\coloneq\CAR(\ell^2(X))$ 
and denote by
\[
\fAloc\coloneq\bigcup_{X\in\cF}\fA_X
\]
the dense subalgebra of local elements of $\fA$. For $A\in\fAloc$, let $\supp(A)$ be 
the smallest $X\in\cF$ such that $A\in\fA_X$.

We denote by $\Theta$ the unique involutive $\ast$-automorphism of $\fA$ satisfying
\[
\Theta(a_x)=-a_x,
\]
for all $x\in\ZZ^d$. The lattice $\ZZ^d$ acts naturally on $\fA$ as a group
$\varphi\coloneq\left(\varphi^x\right)_{x\in\ZZ^d}$ of $\ast$-automorphisms,
$\varphi^x$ mapping $\fA_X$ onto $\fA_{X+x}$. A state $\omega$ on $\fA$ is
\textit{translation-invariant} whenever $\omega\circ\varphi^x=\omega$ for all $x\in\ZZ^d$,
and $\cSI$ denotes the set of all translation-invariant states on $\fA$.

Let \(\omega_\Lambda\) denote the density matrix representing the restriction of the state
\(\omega\in\cSI\) to the finite-dimensional algebra \(\fA_\Lambda\).
The von Neumann entropy of this restriction is
\[
S(\omega_\Lambda):=-\tr(\omega_\Lambda\log \omega_\Lambda).
\]
The limit 
\[
s(\omega)\coloneq\Lim\frac{S(\omega_\Lambda)}{|\Lambda|}
\] 
exists and takes values in $[0,\log 2]$. The number $s(\omega)$ is the 
\textit{specific entropy} of the state $\omega\in\cSI(\fA)$. The entropy map 
\[
\cSI(\fA)\ni\omega\mapsto s(\omega)
\]
is affine and upper semicontinuous. For the proof of these results, 
see~\cite[Section~10.2]{Araki2003}.

An \textit{interaction} is a collection $\Phi=\left(\Phi(X)\right)_{X\in\cF}$, where 
$\Phi(X)$ is a self-adjoint element of $\fA_X$. An interaction 
$\Phi$ is called \textit{even} if $\Theta(\Phi(X))=\Phi(X)$ for all $X\in\cF$, 
and \textit{translation covariant} if $\varphi^x(\Phi(X))=\Phi(X+x)$
for all $x\in\ZZ^d$ and $X\in\cF$.

We will only consider even translation covariant interactions satisfying 
\beq
\|\Phi\|\coloneq\sum_{X\ni0}\|\Phi(X)\|<\infty.
\label{eq:sumab}
\eeq 
The set $\cB$ of such interactions equipped with the norm~\eqref{eq:sumab} is a Banach space. 
In~\cite{Araki2003} Araki and Moriya impose a less restrictive condition than~\eqref{eq:sumab}, 
and we will comment on this point in Section~\ref{sec:remarks}. 
The \textit{specific energy} observable associated to $\Phi\in\cB$ is 
\[
E(\Phi)\coloneq\sum_{X\ni0}\frac{\Phi(X)}{|X|},
\]
and for $\omega\in\cSI$ we set 
\[
\ce_\Phi(\omega)\coloneq\omega(E(\Phi)).
\]
To $\Phi\in\cB$ we associate the \textit{local Hamiltonians} 
\[
H_\Lambda(\Phi)\coloneq\sum_{X\subset\Lambda}\Phi(X).
\]
The following result was established in~\cite{Araki2003}, see in particular 
Section~12.
\begin{theorem}\label{thm:araki-moriya}
Let $\Phi\in\cB$. 
\benr
\item For any $\omega\in\cSI(\fA)$, 
\[
\ce_\Phi(\omega)=\Lim\frac{\omega(H_\Lambda(\Phi))}{|\Lambda|}.
\]
\item The limit 
\[
P(\Phi)=\Lim\tfrac1{|\Lambda|}\log\tr(\e^{-H_\Lambda(\Phi)})
\]
exists and is finite.
\item
\[
P(\Phi)=\sup_{\omega\in\cSI(\fA)}(s(\omega)-\ce_\Phi(\omega)).
\]
\item The set 
\[
\cSeq(\Phi)\coloneq\{\omega\in\cSI(\fA)\mid P(\Phi)=s(\omega)-\ce_\Phi(\omega)\}
\]
is non-empty.
\eenr
\end{theorem}
$\ce_\Phi(\omega)$ is the \textit{specific energy} of the interaction $\Phi$
w.r.t.\;the state $\omega$, and $P(\Phi)$ is the \textit{pressure} of $\Phi$.
Part~\textbf{(iii)} is the \textit{Gibbs variational principle.} The elements of
$\cSeq(\Phi)$ are the \textit{equilibrium states} of $\Phi$.

Much more is known about the thermodynamic formalism of fermionic systems; 
Theorem~\ref{thm:araki-moriya} exposes the minimal list of results we need to 
proceed. 

An interaction $\Phi\in\cB$ generates a dynamics if, for all $A\in\fAloc$, the limit 
\[
\alpha_\Phi^t(A)\coloneq\Lim\e^{\i tH_\Lambda(\Phi)}A\e^{-\i tH_\Lambda(\Phi)}
\]
exists uniformly for $t$ in compact intervals. In this case, $\alpha_\Phi^t$
extends by continuity to $\fA$, and defines a $C^\ast$-dynamics
$\alpha_\Phi=\left(\alpha_\Phi^t\right)_{t\in\RR}$.
We denote by $\cSKMS(\Phi)$ the set of translation-invariant KMS states
for $\alpha_\Phi$ at inverse temperature $\beta=1$. It is a general fact that
\beq
\cSKMS(\Phi)\subset\cSeq(\Phi).
\label{pai}
\eeq
Under the current assumptions this result follows 
from~\cite[Proposition~6.2.17 and~6.2.41]{Bratteli1981} in the case of quantum 
spin systems,  and is implicit in~\cite{Araki2003} in the case of lattice
fermion systems. Since we will not make use of~\eqref{pai}, we do not provide detailed 
reference.

One of the main results of~\cite{Araki2003} is that equality holds in~\eqref{pai} 
under the standard conditions that ensure the existence of $\alpha_\Phi$. 
We proceed to describe them.

To $\Phi\in\cB$ we associate the $\ast$-derivation $\delta_\Phi$ 
acting on $\fAloc$ as  
\[
\delta_\Phi(A)\coloneq\sum_{X\cap\supp(A)\neq\emptyset}\i[\Phi(X),A].
\]
The proof of Proposition~6.2.3 in~\cite{Bratteli1981} gives:
\begin{theorem}\label{thm:br-dyn}
Suppose that $\Phi\in\cB$. Then the derivation $\delta_\Phi$ is norm-closable on $\fAloc$.
Its closure $\bar\delta_\Phi$ generates a $C^\ast$-dynamics on $\fA$ if and only if 
one of the following conditions is satisfied:
\benr
\item $\bar\delta_\Phi$ has a dense set of analytic elements.
\item For any $\kappa\in\RR\setminus\{0\}$, $\left(\Id+\kappa\bar\delta_\Phi\right)\big(\Dom(\bar\delta_\Phi)\big)=\fA$. 
\eenr
If either of these conditions holds, then $\Phi$ generates a dynamics and $\alpha_\Phi^t=\e^{t \bar\delta_\Phi}$. 
\end{theorem}

\begin{remark}\label{rem:notransinv} 
Translation covariance is not required in~\cite[Proposition~6.2.3]{Bratteli1981}; hence, Theorem~\ref{thm:br-dyn} holds as stated for any even interaction $\Phi$ 
satisfying 
$$
\sum_{X\ni x}\|\Phi(X)\|<\infty
$$
for all $x\in\ZZ^d$. 
\end{remark}

\begin{remark} Proposition~6.2.3 in~\cite{Bratteli1981} is proven for quantum spin systems. 
The argument extends directly to lattice fermions.
\end{remark}

The following result holds, see~\cite[Theorems~A and~B]{Araki2003} and Section~\ref{sec:remarks}:
\begin{theorem}\label{eqkms}
Let $\Phi\in\cB$ and suppose that either Condition~\textbf{(i)} or~\textbf{(ii)} of 
Theorem~\ref{thm:br-dyn} holds. Then 
\[
\cSKMS(\Phi)=\cSeq(\Phi).
\]
\end{theorem}

\subsection{Quadratic Interactions}

The following result will be the key input in our proof of Theorem~\ref{thm:main}. 

\bep\label{prop:RegC}
Let $C$ be a regular covariance operator, $h$ the associated one-particle Hamiltonian, and 
$\Phi$ the corresponding interaction.
\benr
\item $\Phi\in\cB$ and Condition~\textup{\textbf{(i)}} of Theorem~\ref{thm:br-dyn} hold.
\item The $C^\ast$-dynamics $\alpha_\Phi$ coincides with the group of Bogoliubov automorphisms
generated by the Hamiltonian $h$.
\item
$$
\cSKMS(\Phi)=\cSeq(\Phi)=\{\omega_C\}.
$$
\item For any $\omega\in\cSI(C)$,
$$
e_\Phi(\omega)=\int_{\TT^d}\frac{\hat h(k)}{1+\e^{\hat h(k)}}\frac{\d k}{(2\pi)^d}.
$$
\eenr
\eep

\proof \textbf{(i)} Since $\hat h\in\cW$, the interaction~\eqref{eq:Phidef} satisfies
$$
\|\Phi\|=\sum_{x\in\ZZ^d}\|\Phi(\{0,x\})\|\le\sum_{x\in\ZZ^d}|h(x,0)|\eqcolon\VERT h\VERT<\infty,
$$
and hence $\Phi\in\cB$. Moreover, if $A\in\fAloc$ with $\supp(A)=X_1$, one has
$$
\delta_\Phi(A)=\sum_{X\cap X_1\neq\emptyset}\i[\Phi(X),A]
=\sum_{x_1\in\ZZ^d}\sum_{y_1\in X_1}\i[\Phi(\{x_1,y_1\}),A]
=\sum_{x_1\in\ZZ^d}A^{(1)}(x_1)
$$
where $A^{(1)}(x_1)\in\fAloc$ and $\supp(A^{(1)}(x_1))=X_1\cup\{x_1\}\eqcolon X_2$. 
Iterating, we get
$$
\delta_\Phi^n(A)=\sum_{x_1,\ldots,x_n\in\ZZ^d}
\sum_{y_n\in X_n}\cdots\sum_{y_1\in X_1}\i[\Phi(\{x_n,y_n\}),\i[\cdots,\i[\Phi(\{x_1,y_1\}),A]\cdots]],
$$
where $X_{k+1}\coloneq X_1\cup\{x_1,\ldots,x_k\}$.
Since $\|\i[\Phi(\{x,y\}),B]\|\le2|h(x,y)|\|B\|$, and for any $Y\in\cF$,
\beq
\sum_{x\in\ZZ^d}\sum_{y\in Y}|h(x,y)|
\le\sum_{x\in\ZZ^d}\sum_{y\in Y}|h(x-y,0)|=|Y|\VERT h\VERT,
\label{eq:hsmall}
\eeq
we get the estimate
$$
\|\delta_\Phi^n(A)\|\le\sum_{x_1,\ldots,x_n\in\ZZ^d}
\sum_{y_n\in X_n}\cdots\sum_{y_1\in X_1}
2^n\left(\prod_{k=1}^n|h(x_k-y_k,0)|\right)\|A\|
\le2^n\frac{(|X_1|+n-1)!}{(|X_1|-1)!}\VERT h\VERT^n\|A\|.
$$
It follows that, for small enough $t\in\RR$,
\beq
\sum_{n\ge0}\frac{t^n}{n!}\|\delta_\Phi^n(A)\|
\le\sum_{n\ge0}\binom{|X_1|+n-1}{|X_1|-1}\left(2t\VERT h\VERT\right)^n\|A\|
\le2^{|X_1|-1}\|A\|\sum_{n\ge0}\left(4t\VERT h\VERT\right)^n<\infty,
\label{eq:analytic}
\eeq
and we conclude that Condition~\textbf{(i)} of Theorem~\ref{thm:br-dyn} 
is satisfied.

\medskip\noindent\textbf{(ii)} By Part~\textbf{(i)} and Theorem~\ref{thm:br-dyn}, 
the interaction $\Phi$ generates a dynamics $\alpha_\Phi$. As already noticed in 
the introduction, the inner $C^\ast$-dynamics generated by the local Hamiltonian 
$H_\Lambda(\Phi)$ coincides with the group of Bogoliubov automorphisms associated 
to the one-particle Hamiltonian $h_\Lambda=\un_\Lambda h\un_\Lambda$,
\beq
\e^{\i tH_\Lambda(\Phi)}A\e^{-\i tH_\Lambda(\Phi)}=\tau_{h_\Lambda}^t(A)
\label{eq:AisB}
\eeq
for any $A\in\fA$ and $t\in\RR$. Since $\Lim h_\Lambda=h$ in the strong topology
of $\cB(\fh)$, one has $\Lim\e^{\i th_\Lambda}=\e^{\i th}$ in the same topology,
for $t\in\RR$. We conclude that for any integer $n$ and any $f_1,\ldots,f_n\in\fh$,
and any $t\in\RR$,
\begin{align*}
\Lim\tau_{h_\Lambda}^t(a^\#(f_1)\cdots a^\#(f_n))
&=\Lim a^\#(\e^{\i th_\Lambda}f_1)\cdots a^\#(\e^{\i th_\Lambda}f_n)\\[4pt]
&=a^\#(\e^{\i th}f_1)\cdots a^\#(\e^{\i th}f_n)=\tau_h^t(a^\#(f_1)\cdots a^\#(f_n)).
\end{align*}
Since elements of the form $a^\#(f_1)\cdots a^\#(f_n)$ are total in $\fA$,
it follows from Relation~\eqref{eq:AisB} that $\tau_h=\alpha_\Phi$.

\medskip\noindent\textbf{(iii)} By Part~\textbf{(i)} and Theorem~\ref{eqkms}, one has 
$\cSeq(\Phi)=\cSKMS(\Phi)$, while Part~\textbf{(ii)} and Theorem~\ref{thm:kms-qf}
give $\cSKMS(\Phi)=\{\omega_C\}$.

\medskip\noindent\textbf{(iv)} By Part~\textbf{(ii)}, for any $\omega\in\cSI(C)$ 
one has, taking Theorem~\ref{thm:kms-qf} into account,
$$
\omega(H_\Lambda(\Phi))=\tr(Ch_\Lambda)=\sum_{x,y\in\Lambda}C(x,y)h(y,x)
=\sum_{x,y\in\Lambda}\int_{\TT^d\times\TT^d}
\frac{\hat h(\eta)}{1+\e^{\hat h(\xi)}}\e^{\i(\xi-\eta)(x-y)}\frac{\d\xi\d\eta}{(2\pi)^{2d}}.
$$
Since, with $\Lambda=\llbracket-L,L\rrbracket^d$,
$$
\Lim\frac1{|\Lambda|}\sum_{x,y\in\Lambda}\e^{\i(\xi-\eta)(x-y)}
=\lim_{L\to\infty}\prod_{j=1}^d\left(\frac1{2L+1}
\sum_{l\in\llbracket-L,L\rrbracket}\e^{\i (\xi_j-\eta_j)l}\right)=(2\pi)^d\delta(\xi-\eta)
$$
in distributional sense, and since $\widehat h$ and $(1+\e^{\hat h})^{-1}\in\cW\subset C(\TT^d)$, 
the result follows.
\qed

 \section{Proof of Theorem \ref{thm:main}}

Define a density matrix in $\fA_\Lambda$ by
\beq
\nu_\Lambda\coloneq\e^{-H_\Lambda(\Phi)-P_\Lambda(\Phi)},
\label{eq:Gibbs}
\eeq
and observe that the relative entropy of the pair $(\omega_\Lambda,\nu_{\Lambda})$ is 
\beq  
\Ent(\omega_\Lambda|\nu_\Lambda)=\tr(\omega_\Lambda(\log\omega_\Lambda-\log\nu_\Lambda))\geq  0,
\label{lr-u}
\eeq
with equality iff $\omega_\Lambda=\nu_\Lambda$.  Since
\beq
\tr(\omega_\Lambda\log\nu_\Lambda)=-\omega_\Lambda(H_\Lambda(\Phi))-P_\Lambda(\Phi),
\label{lr-u-1}
\eeq
dividing~\eqref{lr-u} by $|\Lambda|$ and taking the thermodynamic limit, we deduce 
from~\eqref{lr-u-1} that 
\beq
\ce_\Phi(\omega)+P(\Phi)\geq s(\omega).
\label{lr-u-2}
\eeq
Since $\ce_{\Phi}(\omega)=\ce_{\Phi}(\omega_C)$, and $\omega_C\in\cSeq(\Phi)$, it follows 
from~\eqref{lr-u-2} that 
\[
s(\omega_C)\geq s(\omega).
\]

This proves the Lanford--Robinson Theorem \ref{thm:lr} under our assumptions on $C$, but 
we get more. If $s(\omega)=s(\omega_C)$, then $e_{\Phi}(\omega)=e_{\Phi}(\omega_C)$ implies 
\[ s(\omega)=e_\Phi(\omega) + P(\Phi),\]
and so $\omega\in {\cal S}_\eq(\Phi)$. It then follows from Proposition~\ref{prop:RegC}\textbf{(iii)} that $\omega=\omega_C$, and Theorem \ref{thm:main} is proven.

\subsection{Remarks}
\label{sec:remarks}
 
\noindent\textbf{1.} As mentioned above, the Araki--Moriya thermodynamic formalism
applies to a class of interactions more general than those satisfying
\eqref{eq:sumab}. It is, however, unclear to what extent these broader
assumptions on $\Phi$ would allow one to relax our regularity
condition on $C$. It seems likely that any such extension would require a different technical argument. We plan to address this issue elsewhere.

\medskip\noindent\textbf{2.} The proof of Theorem~\ref{thm:main} is perhaps
deceptively simple; it ultimately relies on deep structural results from the
thermodynamic formalism of lattice fermions. It would be interesting to know
whether a more direct and self-contained proof of this uniqueness result can be
found.

\medskip\noindent\textbf{3.} Our regularity assumption excludes
singular covariance operators, such as those with discontinuous symbol $\widehat{C}$. Although
the entropy maximization principle itself is expected to hold more generally,
extending the uniqueness result to such singular quasi-free states would likely
require a different strategy.

On the other hand, from the thermodynamic perspective, our regularity assumption
is very natural, since it is equivalent to requiring that the associated
interaction $\Phi$ belongs to the class of so-called \textit{small} (or
\textit{physical}) interactions~\cite{Israel1979,Bratteli1981}. The well-known
pathologies that arise with interactions outside this class are discussed 
in~\cite{Israel1979,Bratteli1981}.

\section{Proof of Theorem \ref{thm:weak-gibbs}}
\label{sec:weak-gibbs}

Introducing the decoupled one-particle Hamiltonian
\[ 
\hdec\coloneq h_\Lambda\oplus h_{\Lambda^c},
\]
we note that
$$
\hdec(x,y)=\begin{cases*}
h(x,y)&if $(x,y)\in(\Lambda\times\Lambda)\cup(\Lambda^c\times\Lambda^c)$;\\[4pt]
0&otherwise,
\end{cases*}
$$
and hence, invoking the estimate~\eqref{eq:hsmall}, we deduce that for any $Y\in\cF$,
\beq
\sum_{x\in\ZZ^d}\sum_{y\in Y}|\hdec(x,y)|\le\sum_{x\in\ZZ^d}\sum_{y\in Y}|h(x,y)|=|Y|\VERT h\VERT.
\label{eq:hlamdec}
\eeq
Denote by $\Phidec$ the interaction of the form~\eqref{eq:Phidef}, with $h$
replaced by $\hdec$, \ie
$$
\Phidec(X)=\begin{cases*}
\Phi(X)&if $X\subset\Lambda$ or $X\subset\Lambda^c;$\\[4pt]
0&otherwise.
\end{cases*}
$$
Clearly, $\Phidec$ fails to be translation covariant. However, it is an even interaction 
such that
$$
\sum_{X\ni x}\|\Phidec(X)\|\le\sum_{X\ni x}\|\Phi(X)\|=\|\Phi\|<\infty
$$
for all $x\in\ZZ^d$. Thus, by Remark~\ref{rem:notransinv}, Theorem~\ref{thm:br-dyn} applies. 
Furthermore, by Inequality~\eqref{eq:hlamdec}, the estimate~\eqref{eq:analytic} holds with
$\delta_\Phi$ replaced by $\delta_{\Phidec}$.
We infer that any element of $\fAloc$ is analytic for $\bar\delta_{\Phidec}$, and that
the latter generates a dynamics $\alpha_{\Phidec}$ such that $\alpha_{\Phidec}=\tau_{\hdec}$.
In particular
$$
\alpha_{\Phidec}^t(A)=\e^{\i tH_\Lambda(\Phi)}A\e^{-\i tH_\Lambda(\Phi)},
$$
for any $A\in\fA_\Lambda$, from which it follows that the restriction to $\fA_\Lambda$
of a KMS state for  $\alpha_{\Phidec}$ at $\beta=1$ has the density matrix~\eqref{eq:Gibbs}.

The dynamics $\alpha_\Phi$ and $\alpha_{\Phidec}$ are related by the perturbation 
\[
W_\Lambda(\Phi)=\sum_{\twol{X\cap\Lambda\neq\emptyset}{X\cap\Lambda^c\neq\emptyset}}\Phi(X).
\]
It is well-known that $\Phi\in\cB$ ensures that the \textit{surface energies} $W_\Lambda(\Phi)$ 
are well-defined self-adjoint elements of $\fA$ satisfying  
$$
\Lim\tfrac1{|\Lambda|}\|W_\Lambda(\Phi)\|=0,
$$
see, for example, the proof of~\cite[Theorem 6.2.42]{Bratteli1981}. By construction
\[
W_\Lambda(\Phi)=\sum_{x,y\in\ZZ^d}w_\Lambda(x,y)a_x^\ast a_y,
\]
where $w_\Lambda\coloneq h-\hdec$. Setting $\Lambda_L\coloneq\{x\in\ZZ^d\mid\max_i|x_i|\le L\}$, 
we can write, for $0<R<L$,
\begin{align*}
\frac12\sum_{x,y\in\ZZ^d}|w_{\Lambda_L}(x,y)|
&\le\sum_{(x,y)\in\Lambda_L\times\Lambda_L^c}|h(x,y)|\\[4pt]
&\le\sum_{(x,y)\in(\Lambda_L\setminus\Lambda_{L-R})\times\ZZ^d}|h(x-y,0)|
+\sum_{(x,y)\in\Lambda_{L-R}\times\Lambda_L^c}|h(x-y,0)|\\[4pt]
&\le|\Lambda_L\setminus\Lambda_{L-R}|\VERT h\VERT
+|\Lambda_{L-R}|\sum_{|z|\ge R}|h(z,0)|.
\end{align*}
Dividing by $|\Lambda_L|$, and noticing that $|\Lambda_L\setminus\Lambda_{L-R}|/|\Lambda_L|\to0$
and $|\Lambda_{L-R}|/|\Lambda_L|\to 1$ as $L\to\infty$, yields
$$
\limsup_\Lambda\frac1{2|\Lambda|}\sum_{x,y\in\ZZ^d}|w_{\Lambda}(x,y)|\le\sum_{|z|\ge R}|h(z,0)|.
$$
Since the right-hand side of this inequality vanishes in the limit $R\to\infty$, we 
conclude that 
\beq
\Lim\frac1{|\Lambda|}\sum_{x,y\in\ZZ^d}|w_\Lambda(x,y)|=0.
\label{nir-0}
\eeq
The one-particle Hamiltonians $h$ and $\hdec$ being bounded, the $\fA$-valued maps 
\begin{align*}
\RR\ni t\mapsto\alpha_\Phi^t(W_\Lambda(\Phi))
&=\sum_{x,y\in\ZZ^d}w_\Lambda(x,y)a^\ast(\e^{\i th}\delta_x)a(\e^{\i th}\delta_y),\\[4pt]
\RR\ni t\mapsto\alpha_\Phidec^t(W_\Lambda(\Phi))
&=\sum_{x,y\in\ZZ^d}w_\Lambda(x,y)a^\ast(\e^{\i t\hdec}\delta_x)a(\e^{\i t\hdec}\delta_y),
\end{align*}
have entire analytic extensions such that
\begin{align*}
\|\alpha_\Phi^z(W_\Lambda(\Phi))\|&\leq\e^{2|\Im z|\|h\|}\sum_{x,y\in\ZZ^d}|w_\Lambda(x,y)|,\\[4pt]
\|\alpha_\Phidec^z(W_\Lambda(\Phi))\|&\leq\e^{2|\Im z|\|\hdec\|}\sum_{x,y\in\ZZ^d}|w_\Lambda(x,y)|.
\end{align*}
Relation~\eqref{nir-0} gives  that for all $z\in\CC$, 
\beq
\Lim\tfrac1{|\Lambda|}\|\alpha_\Phi^z(W_\Lambda(\Phi))\|
=\Lim\tfrac1{|\Lambda|}\|\alpha_\Phidec^z(W_{\Lambda}(\Phi))\|=0.
\label{eq:magic}
\eeq
Invoking~\cite[Theorem~3.1]{Jaksic2023a} yields
$$
\e^{-\|W_\Lambda(\Phi)\|-\|\alpha_\Phi^{\i/2}(W_\Lambda(\Phi))\|}\omega_{C,-W_\Lambda(\Phi)}
\le\omega_C
\le\e^{\|W_\Lambda(\Phi)\|+\|\alpha_\Phidec^{\i/2}(W_\Lambda(\Phi))\|}\omega_{C,-W_\Lambda(\Phi)}
$$
where $\omega_{C,-W_\Lambda(\Phi)}$ is a KMS state for $\alpha_\Phidec$ at $\beta=1$. 
Restricting these inequalities to $\fA_\Lambda$, using the previously mentioned fact that
the restriction of $\omega_{C,-W_\Lambda(\Phi)}$ is the state~\eqref{eq:Gibbs}, the result 
now follows from~\eqref{eq:magic}.
\qed

\printbibliography

\end{document}